\documentclass[a4paper]{amsart}

\usepackage[centering]{geometry}
\usepackage{xcolor}

\usepackage[english]{babel}
\usepackage[utf8]{inputenc}
\usepackage[T1]{fontenc}

\usepackage{amsmath}
\usepackage{amssymb}
\usepackage[foot]{amsaddr}

\usepackage{xcolor}

\usepackage{enumerate}
\usepackage{scalerel}

\usepackage{todonotes}
\usepackage[normalem]{ulem}

\usepackage{biblatex}
\usepackage{amsthm}
\usepackage{aliascnt}

\newtheorem{theorem}{Theorem}[section]

\newaliascnt{lemma}{theorem}
\newtheorem{lemma}[lemma]{Lemma}
\aliascntresetthe{lemma}

\newaliascnt{proposition}{theorem}

\aliascntresetthe{proposition}

\newaliascnt{corollary}{theorem}
\newtheorem{corollary}[corollary]{Corollary}
\aliascntresetthe{corollary}

\theoremstyle{definition}

\newaliascnt{definition}{theorem}
\newtheorem{definition}[definition]{Definition}
\aliascntresetthe{definition}

\newaliascnt{example}{theorem}

\aliascntresetthe{example}

\newaliascnt{remark}{theorem}
\newtheorem{remark}[remark]{Remark}
\aliascntresetthe{remark}

\newaliascnt{convention}{theorem}

\aliascntresetthe{convention}

\DeclareMathOperator{\AVaR}{AVaR}

\usepackage[most]{tcolorbox}

\newtcolorbox{notationbox}{
  breakable,
  colframe=black,
  boxrule=0.8pt,
  left=6pt,
  right=6pt,
  top=6pt,
  bottom=6pt
}

\newcommand{\R}{\mathbb{R}}

\newcommand{\FF}{\mathbb{F}}
\newcommand{\E}{\mathbb{E}}

\newcommand{\F}{\mathcal{F}}
\newcommand{\G}{\mathcal{G}}

\newcommand{\cH}{\mathcal{H}}
\newcommand{\simad}{\sim_{\mathrm{ad}}}

\newcommand{\weaklyto}{\rightharpoonup}

\newcommand{\supp}{\textup{supp}}

\newcommand{\id}{\textup{id}}

\newcommand{\law}{\mathcal{L}}

\newcommand{\Pc}{\mathcal P}

\renewcommand{\epsilon}{\varepsilon}
\renewcommand{\subset}{\subseteq}

\renewcommand{\P}{\mathbb{P}}

\newcommand{\lawad}{\law^{\textup{ad}}}

\renewcommand{\phi}{\varphi}

\let\oldmarginpar\marginpar
\renewcommand\marginpar[1]{\-\oldmarginpar[\raggedleft\footnotesize #1]{\raggedright\footnotesize\color{red} #1}}
\usepackage[colorlinks=true, linkcolor=black, citecolor=black, urlcolor=black]{hyperref}
\usepackage{cleveref}

\crefname{theorem}{theorem}{theorems}
\Crefname{theorem}{Theorem}{Theorems}

\crefname{lemma}{lemma}{lemmas}
\Crefname{lemma}{Lemma}{Lemmas}

\crefname{proposition}{proposition}{propositions}
\Crefname{proposition}{Proposition}{Propositions}

\crefname{corollary}{corollary}{corollaries}
\Crefname{corollary}{Corollary}{Corollaries}

\crefname{definition}{definition}{definitions}
\Crefname{definition}{Definition}{Definitions}

\crefname{remark}{remark}{remarks}
\Crefname{remark}{Remark}{Remarks}

\crefname{example}{example}{examples}
\Crefname{example}{Example}{Examples}

\crefname{convention}{convention}{conventions}
\Crefname{convention}{Convention}{Conventions}

\begin{document}

\begin{abstract} 
In static risk measurement, law invariance expresses the principle that the
risk of a position should depend only on its distribution, and not on the
particular probability space on which it is represented. In a dynamic setting,
the same principle leads naturally to adapted law invariance: the risk
assessment should depend only on the probabilistic structure of the financial position 
together with the way information about it is revealed over time.

We show that, for time-consistent risk measures, adapted law invariance is equivalent to a recursive
one-step conditional-law representation. More precisely, assuming Fatou regularity, the one-step risk
evaluations are exactly conditional lifts of static law-invariant risk
measures, and the full dynamic risk measure is obtained by backward
composition of these one-step maps. Convexity and coherence of the dynamic
risk measure are characterized by the corresponding properties of the static
one-step risk measures.

This identifies adapted law invariance as the dynamic counterpart of ordinary
law invariance. It also clarifies the strength of terminal-law invariance, as it appears in
the rigidity theorem of Kupper and Schachermayer: it
does not distinguish risks with the same distribution but different
times of resolution. 
We further obtain an
adapted Kusuoka representation in the coherent case and establish an extension of the Kupper--Schachermayer theorem.

\end{abstract}


\title{Adapted Law Invariance and Time-Consistent Dynamic Risk Measures}
\author{Mathias Beiglböck$^\dagger$}
\address{$^\dagger$Faculty of Mathematics, University of Vienna, Austria}
\email{mathias.beiglboeck@univie.ac.at}

\author{Silvana M. Pesenti$^\ddagger$}
\address{$^\ddagger$Department of Statistical Sciences, University of Toronto, Canada}
\email{silvana.pesenti@utoronto.ca}

\author{Maxime Sylvestre$^\dagger$}
\email{maxime.sylvestre@univie.ac.at}




\thanks{\textit{Acknowledgements.}  The authors are grateful to Walter Schachermayer for his helpful comments and suggestions.}

\maketitle

\section{Introduction}\label{sec:intro}

Law invariance is one of the fundamental structural principles in the theory of
risk measures. A static risk measure
\(
\bar\rho:L^\infty(\Omega,\F,\P)\to\mathbb R
\)
is called law invariant if
\[
\law(X)=\law(\widetilde X)
\quad\Longrightarrow\quad
\bar\rho(X)=\bar\rho(\widetilde X).
\]
Equivalently $\bar \rho (X) = \rho(\law (X))$ for some $\rho: \Pc_b (\R) \to \R$ and by a slight abuse of notation we use the term law invariant risk measure both for $\bar \rho$ and $\rho$. 

Law invariance is a natural requirement both for conceptual and statistical reasons. The
underlying probability space \((\Omega,\F,\P)\) is primarily a representation
device: a random variable \(X:\Omega\to\mathbb R\) represents its distribution
\(\law(X)\), and all probabilistic quantities attached to \(X\), such as
expectations, quantiles, moments, and tail functionals, are determined by this
law. If two random variables have the same distribution, then any distinction
between them which is not visible through their law is a feature of the
particular probability space on which they happen to be represented, rather
than of the risk itself. Moreover, from a statistical point of view, the law is
the only object which can be inferred from observations of the random variable;
one never observes the abstract points of the sample space. Thus a risk measure
which depends on more than the law of \(X\) uses information which is neither
probabilistically intrinsic nor statistically identifiable.

In a dynamic setting the situation is more subtle. The model is no longer just
a probability space, but a filtered probability space
\[
(\Omega,\F,\P,(\F_t)_{t=0}^N),
\]
where the filtration describes the evolution of information. A terminal payoff / liability 
is then not fully described, for dynamic purposes, by its terminal law alone.
Equally important is the way in which information about the payoff is revealed
over time. Two payoffs may have the same terminal distribution while one is
revealed at an early date and the other remains unresolved until maturity.
These two risks have the same ordinary law, but they need not have the same
dynamic risk profile. The appropriate dynamic analogue of law invariance should
therefore retain not only the distribution of the terminal payoff, but also its
information structure.

This leads to the notion of \emph{adapted law invariance}. Roughly speaking, two
random variables, possibly defined on different filtered probability spaces,
are identified if they have the same probabilistic properties together with the
same structure of information revelation. There are various, mathematically equivalent, ways to formalize this intuitions, see \cite{HoKe84, Pf09, BaBePa21, BeiglProbabilistic2024} and we postpone the precise definition to
\Cref{sec:notation-main-results}. Adapted law invariance then means that an object depends
only on this adapted law, and not on any irrelevant features of the filtered
probability space used to realize it.

In dynamic risk theory, the central
structural condition is time consistency. In the sign convention used here,
time consistency means that
\[
R_t(X)=R_t( R_{t+1}(X)),
\qquad t=0,\dots,N-1.
\]
Equivalently, a time-consistent dynamic risk measure is obtained by backward
composition of its one-step restrictions
\[
 S_t:=
 R_t\big|_{L^\infty(\Omega,\F_{t+1},\P)}.
\]
Our  main characterization result shows that, under the natural Fatou regularity
assumption, adapted law invariance of a time-consistent dynamic risk measure is
equivalent to a purely local conditional-law property. More precisely, the
one-step maps must be of the form
\[
S_t(Y)
=
\sigma_t\bigl(\law(Y\mid\F_t)\bigr),
\qquad
Y\in L^\infty(\Omega,\F_{t+1},\P),
\]
where each $\sigma_t:\Pc_b(\mathbb R)\to\mathbb R$ is a static law invariant risk measure.

Thus, once time consistency is imposed, adapted law invariance is characterized
by the requirement that the one-step risk evaluation are static law invariant risk measures depending  only on the
conditional law of the next-period position given the present information. The
full dynamic risk measure is then recovered by backward iteration of these
one-step conditional-law-invariant maps.

A consequence of this characterization  result is an adapted version of the
Kusuoka representation. In the coherent case, with the appropriate Fatou
property, we find that adapted-law-invariant time-consistent dynamic risk
measures are precisely nested conditional Kusuoka operators. Equivalently, the
one-step maps admit representations of the form
\[
\bar\sigma_t(Y)
=
\sup_{\nu\in\mathcal M_t}
\int_{[0,1]}
\AVaR_\alpha(Y\mid\F_t)\,\nu(d\alpha),
\]
where \(\mathcal M_t\) is a nonempty convex weakly closed set of probability
measures on \([0,1]\). The full dynamic risk measure is then obtained by
backward iteration:
\[
\bar\rho_t(X)
=
\bar\sigma_{t}\circ
\bar\sigma_{t+1}\circ\cdots\circ
\bar\sigma_{N-1}(X).
\]
Thus the dynamic analogue of the classical Kusuoka theorem is not a single
static Kusuoka representation at time zero, but a nested representation by
conditional Averaged Value-at-Risk operators.
\medskip

This paper is partially  motivated by the striking rigidity theorem of
Kupper and Schachermayer \cite{KupperSchachermayer2009}. In a standard
filtered probability space, they show that a normalized, relevant,
time-consistent dynamic risk measure whose initial value is law invariant in
the ordinary terminal-law sense must be of entropic type. 
 This theorem is often read as showing that, under ordinary terminal-law
invariance, time consistency leaves essentially only the entropic family.

Our results complement this perspective by isolating the role played by the
notion of law invariance imposed in a filtered setting.

The Kupper--Schachermayer condition identifies terminal payoffs solely through
their ordinary terminal law. Consequently, it abstracts from the time at which
uncertainty is resolved, and in particular does not distinguish whether the
relevant information is revealed early or late. From the viewpoint of filtered
models, this is a strong invariance requirement.
Adapted law invariance is weaker and tailored to the filtered setting:
it treats the temporal resolution of uncertainty as part of the probabilistic
object.
It
identifies two risks only when their laws and their information structures
coincide. Thus the adapted-law-invariant framework retains the natural
statistical and probabilistic motivation of law invariance while avoiding the
loss of information caused by reducing a filtered payoff to its terminal law. 


The Kupper--Schachermayer theorem is formulated in an infinite-horizon setting.
We give a self-contained proof showing that the same rigidity phenomenon is
already present in two periods: for \(N=2\), terminal-law invariance, relevance,
and time consistency imply that the risk measure is of entropic type. This
finite-horizon formulation makes the comparison with our adapted-law-invariant
results transparent.

\medskip

\paragraph{\bf Related literature.}
The axiomatic theory of monetary risk measures starts with coherent risk
measures of Artzner, Delbaen, Eber and Heath \cite{ArDeEbHe99} and with the
subsequent convex theory developed by F{\"o}llmer and Schied
\cite{FoSc02,FollmerSchied2004} and Frittelli and Rosazza Gianin
\cite{FrRoGi02}. The law-invariant subclass is represented by the Kusuoka
theorem for coherent risk measures \cite{Kusuoka2001} and by its convex
extensions and regularity results, in particular
\cite{FrittelliRosazzaGianin2005,JoScTo06,CheriditoLi2008}. Average
Value-at-Risk, or Expected Shortfall, is the basic building block in these
representations; see \cite{AcerbiTasche2002,RockafellarUryasev2002} and, for
later refinements of the Kusuoka representation, \cite{Shapiro2013Kusuoka}.

Dynamic risk measures and time consistency have been studied from several
complementary viewpoints. For coherent and convex conditional risk measures
see, among others,
\cite{Riedel2004,DetlefsenScandolo2005,Artzner2007,CheriditoDelbaenKupper2006,
FollmerPenner2006,delbaen2006,DelbaenPengRosazzaGianin2010,AcciaioPenner2011,
BionNadal2008}. The recursive formulation of time-consistent monetary risk
measures as compositions of one-step maps is standard in the discrete-time
literature; see \cite{CheriditoKupper2011,RuszczynskiShapiro2006,ruszczynski2010risk} and, from a
risk-averse optimisation and control perspective,
\cite{DentchevaRuszczynski2024DynamicRisk}. For unified frameworks and surveys
of the many notions of time consistency in discrete time, see
\cite{BieleckiCialencoPitera2018,BieleckiCialencoPitera2017Survey}. A related
cash-flow perspective, in which the timing of payments is part of the object to
be assessed, is developed in \cite{AcciaioFollmerPenner2012}. Recent related
developments also include dynamic robust risk measures under propagated
uncertainty \cite{MorescoMailhotPesenti2025,coache2026robust}, portfolio-allocation
applications of time-consistent dynamic risk measures
\cite{PesentiJaimungalSaporitoTargino2025}, and risk sharing with dynamic risk measure preferences \cite{tam2026dynamic}.

The interaction between law invariance and dynamic consistency is particularly
delicate. 
Weber \cite{Weber2006} and Kupper--Schachermayer
\cite{KupperSchachermayer2009} show that ordinary (conditional) terminal-law invariance
leads to strong rigidity under dynamic consistency.
A closely related one-step rigidity result is due to Föllmer
\cite[Theorem~3.4]{Follmer2014Spatial}. In an abstract setting with a
single nontrivial sub-\(\sigma\)-field, he shows that a convex law-invariant
risk measure with the Lebesgue property and a suitable consistency property
must be entropic. This result can be viewed as a static/one-step counterpart
of the Kupper--Schachermayer rigidity phenomenon.
The work \cite{Elliott2015JAP} considers conditional law invariance in a Binomial tree setting and a recent continuous-time BSDE
counterpart of Kupper--Schachermayer is \cite{BensaidDumitrescuMatoussiSabbagh2026}. The contribution of the present
paper is to replace terminal-law invariance by adapted law invariance: the law
of the terminal position is not discarded, but the temporal resolution of
uncertainty encoded by the filtration is retained. This leads to a one-step
conditional-law representation and, in the coherent case, to a nested
conditional Kusuoka representation.

\medskip

The paper is organized as follows. In \Cref{sec:notation-main-results} we fix
the notation and state the main results. We first recall static and dynamic
risk measures in the loss convention used throughout the paper, record the
basic locality, time-consistency, reconstruction, and Fatou-regularity facts,
and then introduce adapted laws and adapted-law-invariant functionals. After
restricting to the standard filtered cube, we state the main characterization
theorem: relevant time-consistent adapted-law-invariant risk measures with the
Fatou property are precisely backward compositions of one-step conditional
lifts of relevant static law-invariant risk measures. We also state the
finite-horizon Kupper--Schachermayer rigidity theorem in our convention and
the adapted Kusuoka representation in the coherent case. In
\Cref{sec:dynamic} we prove the main characterization theorem. The proof uses
the reconstruction of the dynamic family from \(R_0\), the one-step
randomization available on the standard cube, and a filtration-preserving
automorphism argument to identify the one-step maps. After the proof of the main characterization, in
\Cref{sec:cash-flow-processes} we record the corresponding formulation for
discounted cash-flow processes under conditional one-step dated cash
additivity and explain why the cash-subadditive case requires a genuinely
process-level extension. In
\Cref{sec:KS-two-period-proof} we give a self-contained proof of the
two-period Kupper--Schachermayer rigidity theorem; as explained in the results
section, this already implies the finite-horizon theorem, and also recovers
the original infinite-horizon result. Finally, in
\Cref{sec:adapted-kusuoka} we proof the adapted Kusuoka representation from
the main characterization theorem and the classical static Kusuoka theorem.


\section{Notation and main results}
\label{sec:notation-main-results}

Throughout the paper all random variables are real-valued and bounded, unless
explicitly stated otherwise. We write \(\Pc_b(\R)\) for the set of Borel
probability measures on \(\R\) with bounded support, equipped with the Borel
structure induced by weak convergence on \(\Pc(\R)\). If \(X\) is a random
variable, then \(\law(X)\) denotes its law. Below we recall the definitions and basic properties of static and dynamic risk measures. For the convenience of the reader we include short proofs. 

\subsection{Static law-invariant risk measures}

We first fix the static notation. 
For this we consider the probability space \((\Omega,\F,\P)\). 
A general risk functional on random variables is
denoted by a capital letter, typically \(R\). Law-level representatives are
denoted by lower-case Greek letters, typically \(\rho\).

\begin{definition}[Static risk measure]
\label{def:static-risk-measure}
A static risk measure is a map \(R:L^\infty(\Omega,\F,\P)\to\R\) such that,
for all \(X,Y\in L^\infty(\Omega,\F,\P)\) and all \(m\in\R\),

\begin{enumerate}[(1)]
\item Monotonicity: if \(X\le Y\), then \(R(X)\le R(Y)\).

\item Cash additivity: \(R(X+m)=R(X)+m\).

\item Normalization: \(R(0)=0\).
\end{enumerate}
\end{definition}

In the loss convention used here, \(X\) is interpreted as a random loss and
\(R(X)\) as its monetary assessment. Thus the usual diversification principle
corresponds to convexity in \(X\).

\begin{definition}[Static properties]
\label{def:static-properties}
Let \(R:L^\infty(\Omega,\F,\P)\to\R\) be a static risk measure.

\begin{enumerate}[(1)]
\item \(R\) is \emph{relevant} if
\(R(\varepsilon\mathbf 1_A)>0\) for every \(\varepsilon>0\) and every
\(A\in\F\) with \(\P(A)>0\).

\item \(R\) is \emph{law invariant} if \(\law(X)=\law(Y)\) implies
\(R(X)=R(Y)\).

\item \(R\) has the \emph{Fatou property} if \(X_n\to X\) a.s. and
\(\sup_n\|X_n\|_\infty<\infty\) imply
\[
R(X)\le\liminf_{n\to\infty}R(X_n).
\]

\item \(R\) is \emph{convex} if
\[
R(\lambda X+(1-\lambda)Y)
\le
\lambda R(X)+(1-\lambda)R(Y)
\]
for all \(X,Y\in L^\infty(\Omega,\F,\P)\) and all \(\lambda\in[0,1]\).

\item \(R\) is \emph{coherent} if it is convex and positively homogeneous,
i.e.\  \(R(aX)=aR(X)\) for all \(a\ge0\).
\end{enumerate}
\end{definition}
\emph{Relevance} is an economically natural property which is technically convenient in the context of time consistent risk measures, see \Cref{lem:reconstruction-from-R0} below. \emph{Convexity} is so ubiquitously assumed as a property of risk measures, that is often taken as part of the definition. By \cite{JoScTo06} every law-invariant convex risk measure has the \emph{Fatou property}. Moreover every law-invariant convex risk measure is relevant, see e.g., corollary 4.59 in \cite{follmer2025book}.

A law-invariant static risk measure is represented by a law-level functional
\(\rho:\Pc_b(\R)\to\R\), in the sense that \(R(X)=\rho(\law(X))\). Conversely,
each such \(\rho\) defines a lift \(\bar\rho\) on every probability space by
\[
\bar\rho(X):=\rho(\law(X)).
\]
We use the term \emph{static law-invariant risk measure} both for the
law-level functional \(\rho\) and for its lift \(\bar\rho\). The properties in
\Cref{def:static-properties} are understood for \(\rho\) through its lift. In
particular, \(\rho\) has the Fatou property iff it is lower semicontinuous on
each compact-support stratum \(\Pc([-M,M])\) with respect to weak convergence.

\subsection{Dynamic risk measures}

We work in finite discrete time. Fix a filtered probability space with horizon
\(N\),
\[
(\Omega,\F,\P,\FF)
=
(\Omega,\F,\P,(\F_t)_{t=0}^N),
\qquad
0=\F_0\subset\F_1\subset\cdots\subset\F_N\subset\F,
\]
where \(0\) denotes the trivial \(\sigma\)-field, up to completion. Throughout we use the notation $L^\infty(\F_t):= L^\infty(\Omega,\F_t,\P)$, $t = 0, \ldots, N$. 

\begin{definition}[Dynamic risk measure]
\label{def:dynamic-risk}
A dynamic risk measure is a family
\(R=(R_t)_{t=0}^N\) of maps
\[
R_t:L^\infty(\F_N)\to L^\infty(\F_t)
\]
such that, for \(t<N\), \(X,Y\in L^\infty(\F_N)\), and
\(m\in L^\infty(\F_t)\), we have

\begin{enumerate}[(1)]
\item Monotonicity: if \(X\le Y\), then \(R_t(X)\le R_t(Y)\) a.s.

\item Cash additivity: \(R_t(X+m)=R_t(X)+m\) a.s.

\item Normalization: \(R_t(0)=0\) a.s.

\item Terminal condition: \(R_N(X)=X\) a.s.
\end{enumerate}
\end{definition}

By cash additivity and normalization, \(R_t(Z)=Z\) for every
\(Z\in L^\infty(\F_t)\), \(t<N\); for \(t=N\) this is the terminal
condition.

\begin{definition}[Dynamic properties]
\label{def:dynamic-properties}
Let \(R=(R_t)_{t=0}^N\) be a dynamic risk measure.

\begin{enumerate}[(1)]
\item \(R\) is \emph{relevant} if \(R_0\) is relevant, that is,
\(R_0(\varepsilon\mathbf 1_A)>0\) for every \(\varepsilon>0\) and every
\(A\in\F_N\) with \(\P(A)>0\).

\item \(R\) is \emph{time consistent} if
\[
R_t(X)=R_t(R_{t+1}(X)),
\qquad t=0,\dots,N-1.
\]

\item \(R\) is \emph{convex} if, for all \(t\), all \(X,Y\), and all
\(\lambda\in L^\infty(\F_t)\) with \(0\le\lambda\le1\),
\[
R_t(\lambda X+(1-\lambda)Y)
\le
\lambda R_t(X)+(1-\lambda)R_t(Y).
\]

\item \(R\) is \emph{coherent} if it is convex and conditionally positively
homogeneous, that is,
\(R_t(\lambda X)=\lambda R_t(X)\) for all
\(\lambda\in L^\infty(\F_t)\) with \(\lambda\ge0\).
\end{enumerate}
\end{definition}

The following elementary locality fact will be used repeatedly.

\begin{lemma}[Locality]
\label{lem:locality}
Let \(S:L^\infty(\F_{t+1})\to L^\infty(\F_t)\) be monotone and cash additive.
Then, for every \(B\in\F_t\) and all \(X,Y\in L^\infty(\F_{t+1})\),
\[
S(\mathbf 1_BX+\mathbf 1_{B^c}Y)
=
\mathbf 1_BS(X)+\mathbf 1_{B^c}S(Y).
\]
Consequently, each \(R_t\) of a dynamic risk measure is local.
\end{lemma}

\begin{proof}
Let \(Z=\mathbf 1_BX+\mathbf 1_{B^c}Y\), and choose
\(c>\|X-Y\|_\infty\). Then
\(X-c\mathbf 1_{B^c}\le Z\le X+c\mathbf 1_{B^c}\). Monotonicity and cash
additivity give
\[
S(X)-c\mathbf 1_{B^c}\le S(Z)\le S(X)+c\mathbf 1_{B^c},
\]
hence \(\mathbf 1_BS(Z)=\mathbf 1_BS(X)\). Interchanging \(B\) and \(B^c\)
gives the claim.
\end{proof}

\begin{lemma}[Equivalent forms of time consistency]
\label{lem:time-consistency-equivalences}
For a dynamic risk measure \(R=(R_t)_{t=0}^N\), the following are equivalent.

\begin{enumerate}[(i)]
\item Recursive time consistency:
\(
R_t(X)=R_t(R_{t+1}(X)),
 t=0,\dots,N-1.
\)

\item Semigroup time consistency:
\(
R_s(X)=R_s(R_t(X)),
 0\le s\le t\le N.
\)

\item Strong time consistency:
\(
R_{t+1}(X)\le R_{t+1}(Y)
\quad\Longrightarrow\quad
R_t(X)\le R_t(Y),
 t=0,\dots,N-1.
\)

\item Equality consistency:
\(
R_{t+1}(X)=R_{t+1}(Y)
\ \Longrightarrow\ 
R_t(X)=R_t(Y),
\ t=0,\dots,N-1.
\)
\end{enumerate}

If, in addition, \(R\) is relevant, these conditions are also equivalent to
the time-zero condition
\[
R_0(X)=R_0(R_t(X)),
\qquad t=0,\dots,N.
\]
\end{lemma}

\begin{proof}
The equivalence of (i) and (ii) follows by iteration; (ii) implies (i) by
choosing \(s=t\) and replacing \(t\) by \(t+1\).

Assume (i). If \(R_{t+1}(X)\le R_{t+1}(Y)\), then monotonicity of \(R_t\)
gives
\[
R_t(R_{t+1}(X))\le R_t(R_{t+1}(Y)),
\]
hence \(R_t(X)\le R_t(Y)\). Thus (i) implies (iii). Clearly, (iii) implies
(iv). Conversely, if (iv) holds, apply it to \(Y:=R_{t+1}(X)\). Since
\(R_{t+1}(Y)=Y=R_{t+1}(X)\), one obtains
\(R_t(X)=R_t(R_{t+1}(X))\). Hence (i)--(iv) are equivalent.

Under any of these conditions, the time-zero condition follows from (ii) with
\(s=0\). Conversely, assume relevance and the time-zero condition. Fix
\(t<N\), set \(U:=R_t(X)\) and \(W:=R_t(R_{t+1}(X))\). We prove \(U=W\).

Let \(A\in\F_t\), \(h\in L^\infty(\F_t)\), and
\(Z:=\mathbf 1_AX+\mathbf 1_{A^c}R_{t+1}(X)\). Locality gives
\[
R_t(Z+h)=\mathbf 1_AU+\mathbf 1_{A^c}W+h,
\qquad
R_{t+1}(Z+h)=R_{t+1}(X)+h.
\]
Using the time-zero condition at times \(t,t+1,t\), respectively, yields
\[
R_0(\mathbf 1_AU+\mathbf 1_{A^c}W+h)=R_0(W+h).
\]
If \(U>W\) on a set of positive probability, choose \(A\in\F_t\) and
\(\varepsilon>0\) such that \(U\ge W+\varepsilon\) on \(A\). With \(h=-W\) we
get \(R_0(\mathbf 1_A(U-W))=0\), contradicting monotonicity and relevance.
Thus \(U\le W\).

For the reverse inequality, let
\(\widetilde Z:=\mathbf 1_AR_t(X)+\mathbf 1_{A^c}X\). Locality gives
\(R_t(\widetilde Z+h)=U+h\), while
\[
R_t(R_{t+1}(\widetilde Z+h))
=
\mathbf 1_A(U+h)+\mathbf 1_{A^c}(W+h).
\]
The time-zero condition implies
\[
R_0(U+h)=R_0(\mathbf 1_AU+\mathbf 1_{A^c}W+h).
\]
If \(W>U\) on a set \(B\in\F_t\) of positive probability, take
\(A=B^c\), \(h=-U\), and argue as above. Hence \(W\le U\), and therefore
\(U=W\).
\end{proof}

For a time-consistent dynamic risk measure we define its one-step maps by
\[
S_t:=R_t|_{L^\infty(\F_{t+1})}:
L^\infty(\F_{t+1})\to L^\infty(\F_t),
\qquad t=0,\dots,N-1.
\]
Then \(R_N=\id\) and, for \(t<N\), $R_t$ is given via the backward recursion
\[
R_t=S_t\circ S_{t+1}\circ\cdots\circ S_{N-1}.
\]
The equivalence of time consistency and the above backward recursion has been established in \cite{CheriditoDelbaenKupper2006, ruszczynski2010risk}. Indeed, time-inconsistency of conditional risk measures mapping $L^\infty(\F_N)$ to $L^\infty(\F_t)$, have been documented, see e.g., \cite{cheridito2009VaR} for conditional Value-at-Risk and \cite{Artzner2007} for conditional Average Value-at-Risk or Expected Shortfall.

\begin{lemma}[Reconstruction from the initial functional]
\label{lem:reconstruction-from-R0}
Let \(R=(R_t)_{t=0}^N\) be time consistent and relevant. Then \(R\) is uniquely
determined by \(R_0\). More precisely, for every \(t=0,\dots,N\) and every
\(X\in L^\infty(\F_N)\),
\[
R_t(X)
=
\operatorname*{ess\,inf}
\Bigl\{
m\in L^\infty(\F_t):
R_0(\mathbf 1_A(X-m))\le0
\text{ for every }A\in\F_t
\Bigr\}.
\]
\end{lemma}

\begin{proof}
Put \(V:=R_t(X)\). For every \(A\in\F_t\), locality and cash additivity give
\(R_t(\mathbf 1_A(X-V))=0\). Hence, by time consistency,
\(R_0(\mathbf 1_A(X-V))=0\), so \(V\) belongs to the set on the right.

Conversely, let \(m\in L^\infty(\F_t)\) belong to that set. If
\(V>m\) on a set of positive probability, then for some
\(A\in\F_t\) and \(\varepsilon>0\) one has \(V\ge m+\varepsilon\) on \(A\).
Locality gives
\[
R_t(\mathbf 1_A(X-m))=\mathbf 1_A(V-m)\ge\varepsilon\mathbf 1_A.
\]
Thus, by time consistency, monotonicity, and relevance,
\[
R_0(\mathbf 1_A(X-m))
=
R_0(R_t(\mathbf 1_A(X-m)))
\ge
R_0(\varepsilon\mathbf 1_A)>0,
\]
contradicting the defining inequality. Hence \(V\le m\). Since \(V\) itself is
admissible, it is the essential infimum.
\end{proof}

\begin{remark}[On the opposite sign convention]
\label{rem:reconstruction-sign}
The formally dual formula
\[
R_t(X)
=
\operatorname*{ess\,sup}
\{m\in L^\infty(\F_t):
R_0(\mathbf 1_A(X-m))\ge0\text{ for every }A\in\F_t\}
\]
is not valid under the present axioms. For instance, if \(R_0=\operatorname{ess\,sup}\),
then \(R_0(-c\mathbf 1_A)=0\) whenever \(A^c\) is non-null, so the displayed
essential supremum may be too large. The essential-infimum formula above is
the one compatible with the loss convention and the usual relevance assumption.
\end{remark}

\subsection{Adapted laws and adapted law invariance}
\label{subsec:adapted-laws}

In the static setting, the probability space is mainly a representation device:
a random variable \(X\) represents its law \(\law(X)\), and law invariance
expresses the principle that no further structure of the chosen probability
space should matter. In the filtered setting, the corresponding object is not
only the terminal law of \(X\), but also the way in which information about
\(X\) is revealed over time. At time \(0\) one only knows the initial
distributional information; at time \(N\), since \(X\) is \(\F_N\)-measurable,
the outcome is known. The intermediate \(\sigma\)-fields describe the gradual
resolution of uncertainty.

Thus a filtered random variable
\[
\mathbf X=(\Omega,\F,\P,(\F_t)_{t=0}^N,X)
\]
should be identified with another filtered random variable (on a possibly different probability space)
\[
\widetilde{\mathbf X}
=(\widetilde\Omega,\widetilde\F,\widetilde\P,
(\widetilde\F_t)_{t=0}^N,\widetilde X)
\]
if they have the same intrinsic filtered probabilistic structure. There are
several equivalent ways to make this precise. 

\begin{enumerate}
    \item 
    Hoover--Keisler
\cite{HoKe84} formalize  from a model-theoretic point of view what it means that $X, \widetilde X$ have the same probabilistic properties.
\item
In optimal transport one may say that the adapted Wasserstein distance between $X$ and $\widetilde X$
 is zero \cite{Pf09, BaBePa21}. 
 \item Equivalently, there
exist real-valued Markov processes \(M=(M_t)_{t=0}^N\) and
\(\widetilde M=(\widetilde M_t)_{t=0}^N\) on \((\Omega,\F,\P,(\F_t)_{t=0}^N)\) and \( (\widetilde\Omega,\widetilde\F,\widetilde\P,
(\widetilde\F_t)_{t=0}^N)\), resp.\   with \(M_N=X\),
\(\widetilde M_N=\widetilde X\), and such that \(M\) and \(\widetilde M\)
have the same joint law, see \cite{BeiglProbabilistic2024}. 
\end{enumerate}
We use the
equivalent formulation in terms of iterated conditional laws given in
\cite{BePaSc25}.

For \(k\ge0\), write \(\Pc_b^0(\R):=\R\) and
\(\Pc_b^{k+1}(\R):=\Pc_b(\Pc_b^k(\R))\), with the usual bounded-support Borel
structure. For \(X\in L^\infty(\F_N)\), its adapted law is the
iterated conditional law
\begin{align}
\lawad(X)
:=
\law\bigl(
\law(
\cdots
\law(
\law(X\mid\F_{N-1})
\mid\F_{N-2})
\cdots
\mid\F_1)
\bigr)
\in\Pc_b^N(\R).
\label{eq:adapted-law-iterated}
\end{align}
Here \eqref{eq:adapted-law-iterated} is read recursively: after the first step,
\(\law(X\mid\F_{N-1})\) is a \(\Pc_b(\R)\)-valued random variable, and each
subsequent conditional law is the regular conditional law of the preceding
state-valued random variable. Since \(\F_0\) is trivial, the outermost law is a
deterministic element of \(\Pc_b^N(\R)\).

\begin{definition}[Adapted law]
\label{def:adapted-law}
Two bounded terminal random variables \(X\) and \(\widetilde X\), possibly
defined on different filtered probability spaces with the same horizon, have
the same adapted law, written \(X\simad\widetilde X\), if
\[
\lawad(X)=\lawad(\widetilde X).
\]
\end{definition}

The adapted law refines the ordinary terminal law.\footnote{There is a
canonical terminal-marginal map
\(\operatorname{ter}_N:\Pc_b^N(\R)\to\Pc_b(\R)\) such that
\(\operatorname{ter}_N(\lawad(X))=\law(X)\), indeed
\(\operatorname{ter}_N\) is a repeated application of the intensity operator.}
Hence equality of adapted laws implies equality of terminal laws. The converse
is false: two payoffs may have the same terminal distribution while one is
revealed early and the other only at maturity.

\begin{definition}[Adapted-law-invariant functional]
\label{def:ad-li-functional}
A real-valued functional \(H\), assigning to each filtered random variable
\(\mathbf X=(\Omega,\F,\P,(\F_t)_{t=0}^N,X)\) a real number \(H(X)\), is
\emph{adapted-law invariant} if
\[
X\simad\widetilde X
\quad\Longrightarrow\quad
H(X)=H(\widetilde X).
\]
Equivalently, \(H\) factors through \(\lawad\).
\end{definition}

\begin{remark}
\label{rem:standardEnough}
While one can consider equivalence in adapted law for random variables on
different filtered probability spaces, this is sometimes cumbersome\footnote{
E.g.\ a functional on all filtered random variables is then necessarily a
proper class (in the sense of set theory) and there is no canonical way of
defining measurability of such functionals.} and can be avoided: by
\cite{BePaSc25}, every adapted law can be represented on an a priori fixed
probability space provided that this space is rich enough, in particular if
this space is a standard filtered probability space.
\end{remark}

{\bf Convention:} From now on we work on a standard filtered probability space,
which we take wlog as the standard filtered cube
\[
\mathbb U_N
:=
([0,1]^N,\mathcal B([0,1]^N),\lambda^N,(\F_t^N)_{t=0}^N),
\qquad
\F_t^N:=\sigma(\pi_1,\dots,\pi_t),
\]
with \(\F_0^N=0\) and $\pi_t :[0,1]^N \to [0,1]$ the projection onto the $t$-th variable. We suppress the superscript \(N\) and write simply
\(\F_t\).

\subsection{Main characterization theorem}
\label{subsec:main-characterization}

In view of \Cref{lem:reconstruction-from-R0}, adapted law invariance is the
natural invariance notion for dynamic risk measures. A relevant
time-consistent dynamic risk measure \(R=(R_t)_{t=0}^N\) (defined on the
standard filtered cube \(\mathbb U_N\)) is called \emph{adapted-law invariant} if
its initial functional \(R_0\) is adapted-law invariant, that is,
\[
X\simad\widetilde X
\quad\Longrightarrow\quad
R_0(X)=R_0(\widetilde X).
\]
Indeed, once relevance and time consistency are imposed, the whole family
\(R\) is reconstructed from \(R_0\); thus no additional invariance condition at
later times has to be specified separately.

\begin{theorem}[Main characterization of relevant time-consistent risk measures]
\label{thm:main-dynamic-characterization}
\label{thm:main-standard-cube}
Let \(R=(R_t)_{t=0}^N\) be a relevant time-consistent dynamic risk measure on
\(\mathbb U_N\), and assume that \(R_0\) has the Fatou property. Then the
following are equivalent.

\begin{enumerate}[(i)]
\item \(R\) is adapted-law invariant, that is, \(R_0\) is adapted-law invariant.

\item There are unique relevant static law-invariant risk measures
\(\rho_t:\Pc_b(\R)\to\R\), \(t=0,\dots,N-1\), with the Fatou property, such
that, for every \(Y\in L^\infty(\F_{t+1})\),
\begin{align}
S_t(Y)=\rho_t(\law(Y\mid\F_t))
\qquad\text{a.s.}
\label{eq:main-one-step-representation}
\end{align}
Put otherwise, for every \(X\in L^\infty(\F_N)\),
\begin{align}
R_N(X)&=X,
&
R_t(X)&=\rho_t(\law(R_{t+1}(X)\mid\F_t)),
\quad t=N-1,\dots,0.
\label{eq:main-backward-recursion}
\end{align}
\end{enumerate}

Under these equivalent conditions, \(R\) is convex if and only if each
\(\rho_t\) is convex, and \(R\) is coherent if and only if each \(\rho_t\) is
coherent.

Conversely, any family \(\rho_0,\dots,\rho_{N-1}\) of relevant static
law-invariant risk measures with the Fatou property defines, by
\eqref{eq:main-backward-recursion}, a relevant time-consistent dynamic risk
measure on \(\mathbb U_N\). Its initial functional has the Fatou property and
is adapted-law invariant. If the $\rho_t$'s are convex, respectively
coherent, then so is the resulting dynamic risk measure.
\end{theorem}

\begin{remark}[Role of the standard filtered space]
\label{rem:standard-filtered-space-role}
The assumption that we work on the standard filtered cube is used only as a
richness assumption. What is needed is that, at each date \(t<N\), the model
contains fresh randomness after time \(t\): equivalently, the probability space
is conditionally atomless over \(\F_t\). This allows us to realize arbitrary
\(\F_t\)-measurable conditional laws between times \(t\) and \(t+1\). In this
sense the standard filtered cube plays exactly the same role in the dynamic
setting as an atomless probability space plays in the static theory of
law-invariant risk measures.
\end{remark}

\subsection{The Kupper--Schachermayer rigidity theorem revisited}
\label{subsec:KS-revisited}

We now compare adapted law invariance with the stronger terminal-law
invariance used by Kupper and Schachermayer
\cite{KupperSchachermayer2009}. We keep the loss convention of the present
paper; Kupper and Schachermayer use the opposite sign convention.

\begin{definition}[Terminal-law invariance]
\label{def:terminal-law-invariance}
A dynamic risk measure \(R=(R_t)_{t=0}^N\) is called
\emph{terminal-law invariant}, or law invariant in the
Kupper--Schachermayer sense, if \(R_0\) is law invariant with respect to the
ordinary terminal law:
\[
\law(X)=\law(\widetilde X)\quad\Longrightarrow\quad R_0(X)=R_0(\widetilde X).
\]
\end{definition}

For \(\gamma\in\R\setminus\{0\}\), we define the conditional entropic risk measure 
\begin{align}
\mathcal E_t^\gamma(X)
:=
\frac1\gamma\log\E[e^{\gamma X}\mid\F_t],
\label{eq:KS-conditional-entropic}
\end{align}
with limiting conventions
\begin{align}
\mathcal E_t^0(X):=\E[X\mid\F_t],
\qquad
\mathcal E_t^\infty(X):=\operatorname{ess\,sup}(X\mid\F_t).
\label{eq:KS-conditional-entropic-limits}
\end{align}

\begin{theorem}[Kupper--Schachermayer rigidity, finite-horizon form]
\label{thm:KS-finite-rigidity}
Let \(N\ge2\), and let \(R=(R_t)_{t=0}^N\) be a relevant time-consistent
dynamic risk measure on the standard filtered cube \(\mathbb U_N\). If \(R\)
is terminal-law invariant, then there exists
\(\gamma\in(-\infty,\infty]\) such that
\[
R_t(X)=\mathcal E_t^\gamma(X),
\qquad t=0,\dots,N.
\]
Conversely, each family \(R_t=\mathcal E_t^\gamma\),
\(\gamma\in(-\infty,\infty]\), is relevant, time consistent, and
terminal-law invariant. Moreover, \(R\) is convex if and only if
\(\gamma\in[0,\infty]\), and coherent if and only if
\(\gamma\in\{0,\infty\}\).
\end{theorem}


The original Kupper--Schachermayer theorem is formulated on the infinite
standard filtered cube, with its canonical coordinate filtration, and its proof
uses the infinite-horizon structure, in particular through a discrete
Skorohod-embedding argument. The argument below shows that the same rigidity
mechanism is already present in two periods.

Indeed, suppose first that \(R\) is given on the finite standard cube
\(\mathbb U_N\), \(N\ge2\). Consider the coarser two-step filtration
\[
\G_0:=\F_0,\qquad
\G_1:=\F_1,\qquad
\G_2:=\F_N,
\]
and define
\[
\widehat R_0:=R_0,\qquad
\widehat R_1:=R_1,\qquad
\widehat R_2:=\id .
\]
The filtered space with filtration \((\G_0,\G_1,\G_2)\) is isomorphic to the
standard two-period cube. Moreover, terminal-law invariance, relevance,
monotonicity and cash additivity are inherited by \(\widehat R\), and the
time-zero consistency condition
\[
\widehat R_0(X)=\widehat R_0(\widehat R_1(X))
\]
follows from the time consistency of \(R\). By
\Cref{lem:time-consistency-equivalences}, \(\widehat R\) is therefore a
relevant time-consistent two-period dynamic risk measure. The two-period
rigidity theorem then gives
\begin{align}
R_0(X)=\mathcal E_0^\gamma(X),
\qquad
X\in L^\infty(\F_N),
\label{eq:KS-R0-entropic-from-two-period}
\end{align}
for some \(\gamma\in(-\infty,\infty]\). Once
\eqref{eq:KS-R0-entropic-from-two-period} holds, the reconstruction lemma
yields \(R_t=\mathcal E_t^\gamma\) for every \(t=0,\dots,N\). Thus \(N=2\)
rigorously implies the finite-horizon theorem for all \(N\ge2\). The same
argument yields that the infinite-horizon Kupper--Schachermayer rigidity
theorem is a consequence of the two-period theorem.

Let us also mention that, under additional convexity, strong sensitivity and
Lebesgue-type regularity assumptions, the two-period entropic conclusion is
already contained in Föllmer's one-step rigidity theorem
\cite[Theorem~3.4]{Follmer2014Spatial}. Indeed, in the loss convention used
here, applying that theorem to
\(\widehat\rho(Z):=R_0(-Z)\) and
\(\widehat\rho_0(Z):=R_1(-Z)\) turns time consistency into Föllmer's
\(\mathcal F_1\)-consistency condition. Our proof of \Cref{thm:KS-finite-rigidity} is self-contained,
does not impose convexity or Fatou/Lebesgue regularity, and also covers the
nonconvex and endpoint cases appearing in \(\gamma\in(-\infty,\infty]\).


\medskip

Law invariance in the sense of Kupper--Schachermayer requires that, at time
zero, the risk of a position depend only on its terminal distribution. It
therefore abstracts from the time at which uncertainty is resolved and from the
information mechanism through which it is resolved. In a setup where the flow of information is modeled in terms of a filtration this is
a strong invariance requirement.

For instance, on \(\mathbb U_2\), let
\[
X:=\mathbf 1_{\{\pi_1\le p\}},
\qquad
Y:=\mathbf 1_{\{\pi_2\le p\}},
\qquad p\in(0,1).
\]
where $\pi_t :[0,1]^2 \to [0,1]$ is the projection onto the $t$-th variable.
Then \(X\) and \(Y\) have the same terminal law. But \(X\) is already known at
time \(1\), while \(Y\) is still unresolved at time \(1\). Thus their adapted laws
are different. Terminal-law invariance forces \(R_0(X)=R_0(Y)\); adapted law
invariance does not. 
The Kupper--Schachermayer rigidity theorem shows that, once this abstraction
from the information structure is combined with relevance and time consistency
on a rich filtered space, the entropic family is singled out.

\subsection{The adapted Kusuoka representation}
\label{subsec:adapted-kusuoka-results}

We now record the coherent consequence of the main characterization theorem.
For \(\mu\in\Pc_b(\R)\), consider the quantile function
\[
q_\mu(u):=\inf\{x\in\R:\mu((-\infty,x])\ge u\},
\qquad u\in(0,1),
\]
and define, for \(\alpha\in[0,1)\), the Average Value-at-Risk,
\begin{align}
\AVaR_\alpha(\mu)
&:=
\frac{1}{1-\alpha}\int_\alpha^1 q_\mu(u)\,du,
&
\AVaR_1(\mu)
&:=\sup\supp(\mu).
\label{eq:AVaR-law-definition}
\end{align}
For \(Y\in L^\infty(\F_{t+1})\), we write
\begin{align}
\AVaR_\alpha(Y\mid\F_t)
:=
\AVaR_\alpha(\law(Y\mid\F_t)).
\label{eq:conditional-AVaR-definition}
\end{align}

\begin{corollary}[Adapted Kusuoka representation]
\label{cor:adapted-kusuoka}
Let \(R=(R_t)_{t=0}^N\) be a relevant time-consistent dynamic risk measure on
the standard filtered cube \(\mathbb U_N\), and assume that \(R_0\) has the
Fatou property. Then the following are equivalent.

\begin{enumerate}[(i)]
\item \(R\) is adapted-law invariant and coherent.

\item For every \(t=0,\dots,N-1\), there exists a nonempty convex weakly closed
set \(\mathcal M_t\subset\Pc([0,1])\) such that, for every
\(Y\in L^\infty(\F_{t+1})\),
\begin{align}
S_t(Y)
=
\sup_{\nu\in\mathcal M_t}
\int_{[0,1]}\AVaR_\alpha(Y\mid\F_t)\,\nu(d\alpha)
\qquad\text{a.s.}
\label{eq:adapted-kusuoka-one-step}
\end{align}
\end{enumerate}

Put otherwise, the dynamic risk measure is the
nested conditional Kusuoka operator
\begin{align}
R_N(X)&=X,
&
R_t(X)
&=
\sup_{\nu\in\mathcal M_t}
\int_{[0,1]}
\AVaR_\alpha(R_{t+1}(X)\mid\F_t)\,\nu(d\alpha),
\quad t=N-1,\dots,0.
\label{eq:adapted-kusuoka-nested}
\end{align}

Conversely, every choice of nonempty convex weakly closed sets
\(\mathcal M_t\subset\Pc([0,1])\), \(t=0,\dots,N-1\), defines by
\eqref{eq:adapted-kusuoka-nested} a relevant time-consistent dynamic risk
measure whose initial functional has the Fatou property and which is
adapted-law invariant and coherent.
\end{corollary}

\begin{remark}
\label{rem:kusuoka-nested-not-static}
The representation is nested. It is not, in general, a single static Kusuoka
representation of \(R_0\). The dynamic structure is carried by the successive
conditional laws of \(R_{t+1}(X)\) given \(\F_t\). The one-step risk measures
\(\rho_t\) are uniquely determined by \(R\), but the representing sets
\(\mathcal M_t\) need not be unique.
\end{remark}

\section{Proof of the main characterization}
\label{sec:dynamic}

We prove the results announced in \Cref{subsec:main-characterization}. We use
the locality and reconstruction facts from
\Cref{lem:locality,lem:reconstruction-from-R0}.
Throughout this section the underlying filtered probability space is the
standard filtered cube \(\mathbb U_N\), unless explicitly stated otherwise, and
\(S_t:=R_t|_{L^\infty(\F_{t+1})}\).

Before we give the actual proof, we collect some auxiliary lemmas.

\begin{definition}[One-step Fatou property]
\label{def:one-step-fatou}
A time-consistent dynamic risk measure \(R=(R_t)_{t=0}^N\) has the
\emph{one-step Fatou property} if each \(S_t\) has the Fatou property on
\(L^\infty(\F_{t+1})\): whenever \(Y_n,Y\in L^\infty(\F_{t+1})\),
\(\sup_n\|Y_n\|_\infty<\infty\), and \(Y_n\to Y\) a.s.,
\[
S_t(Y)\le\liminf_{n\to\infty}S_t(Y_n)
\qquad\text{a.s.}
\]
\end{definition}

\begin{lemma}[Fatou regularity propagates to one-step maps]
\label{lem:R0-fatou-implies-one-step-fatou}
Let \(R=(R_t)_{t=0}^N\) be time consistent and relevant. If \(R_0\) has the
Fatou property, then every \(R_t\), and in particular every one-step map
\(S_t\), has the Fatou property.
\end{lemma}

\begin{proof}
It is enough to prove continuity from below. Let \(X_n\uparrow X\) a.s.\ with a
common bound, and set \(V_n:=R_t(X_n)\), \(V:=\lim_n V_n\). By monotonicity,
\(V\le R_t(X)\). If the inequality is strict on some \(A\in\F_t\), then
\(R_t(X)\ge V+\varepsilon\) on a smaller non-null \(A\), for some
\(\varepsilon>0\). Since \(V_n\le V\), locality gives
\[
R_t(\mathbf 1_A(X_n-V))\le R_t(\mathbf 1_A(X_n-V_n))=0,
\]
hence \(R_0(\mathbf 1_A(X_n-V))\le0\). The variables
\(\mathbf 1_A(X_n-V)\) increase a.s. to \(\mathbf 1_A(X-V)\), so the Fatou
property of \(R_0\) gives \(R_0(\mathbf 1_A(X-V))\le0\). But
\(R_t(\mathbf 1_A(X-V))=\mathbf 1_A(R_t(X)-V)\ge\varepsilon\mathbf 1_A\),
and time consistency, monotonicity, and relevance yield
\(R_0(\mathbf 1_A(X-V))>0\), a contradiction. Therefore
\(R_t(X_n)\uparrow R_t(X)\).

Now let \(X_n\to X\) a.s. with a common bound, and put
\(Y_k:=\inf_{n\ge k}X_n\). Then \(Y_k\uparrow X\), and
\(Y_k\le X_n\) for \(n\ge k\). By continuity from below and monotonicity,
\[
R_t(X)=\lim_k R_t(Y_k)
\le \liminf_{n\to\infty}R_t(X_n).
\]
Thus \(R_t\) has the Fatou property. Restricting to
\(L^\infty(\F_{t+1})\) gives the one-step Fatou property of \(S_t\).
\end{proof}

As in basic
measure theory, Fatou implies monotone convergence.

\begin{lemma}[Continuity from below]
\label{lem:main-continuity-from-below}
Let \(S:L^\infty(\G)\to L^\infty(\cH)\) be monotone and have the Fatou
property. If \(X_n\uparrow X\) a.s. and \(\sup_n\|X_n\|_\infty<\infty\), then
\[
S(X_n)\uparrow S(X)
\qquad\text{a.s.}
\]
The same holds for the lift of a static law-invariant risk measure with the
Fatou property.
\end{lemma}

\begin{proof}
By monotonicity, \(S(X_n)\uparrow V\le S(X)\). The Fatou property gives
\(S(X)\le\liminf_n S(X_n)=V\).
\end{proof}

\begin{lemma}[Conditional lifts]
\label{lem:main-conditional-lifts}
Let \(\rho:\Pc_b(\R)\to\R\) be a Borel static law-invariant risk measure and
set
\[
S_t^\rho(Y):=\rho(\law(Y\mid\F_t)),
\qquad Y\in L^\infty(\F_{t+1}).
\]
Then \(S_t^\rho\) is normalized, monotone, and cash additive. If \(\rho\) has
the Fatou property, is convex, or is coherent, then \(S_t^\rho\) has the
corresponding one-step Fatou, conditional convexity, or conditional coherence
property. Conversely, these one-step properties imply the corresponding
properties of \(\rho\) by testing on one-step spaces with trivial current
\(\sigma\)-field.
\end{lemma}

\begin{proof}
The assertions are pointwise in the current state. For cash additivity, if
\(m\in L^\infty(\F_t)\), then
\[
\law(Y+m\mid\F_t)(\omega)
=
\big(\tau_{m(\omega)}\big)_\#\law(Y\mid\F_t)(\omega),
\]
and cash additivity of \(\rho\) gives \(S_t^\rho(Y+m)=S_t^\rho(Y)+m\).
Monotonicity, convexity, and positive homogeneity follow by applying the
corresponding static property to conditional joint laws.

For the Fatou property, take a regular conditional law of
\(\big(Y,Y_1,Y_2,\dots\big)\) given \(\F_t\). If \(Y_n\to Y\) a.s. and the
sequence is uniformly bounded, then for a.e. \(\omega\),
\(\law(Y_n\mid\F_t)(\omega)\) converges weakly to
\(\law(Y\mid\F_t)(\omega)\) on a common compact support. Lower semicontinuity
of \(\rho\) on bounded-support strata gives the claim.
\end{proof}

For \(t\ge1\), call a measure-preserving automorphism \(T\) of the first
\(t\) coordinates \emph{adapted} if
\(T^{-1}\sigma(\pi_1,\dots,\pi_s)=\sigma(\pi_1,\dots,\pi_s)\) for
\(s=0,\dots,t\). We extend such \(T\) to the full cube by leaving the
coordinates \(t+1,\dots,N\) fixed, and write \(U_TZ:=Z\circ T\).

\begin{lemma}[Equivariance under adapted automorphisms]
\label{lem:adapted-automorphism-equivariance}
Let \(R\) be relevant and time consistent, and assume that \(R_0\) is
adapted-law invariant. If \(T\) is an adapted automorphism of the first
\(t\) coordinates, then
\[
S_t(U_TY)=U_TS_t(Y),
\qquad Y\in L^\infty(\F_{t+1}).
\]
\end{lemma}

\begin{proof}
Since \(T\) preserves the filtration, \(U_TZ\simad Z\) for every terminal
random variable \(Z\). Hence \(R_0(U_TZ)=R_0(Z)\).

We use the reconstruction formula. Let \(m\in L^\infty(\F_t)\) and
\(A\in\F_t\). Put \(n:=m\circ T^{-1}\), and choose \(B\in\F_t\) with
\(A=T^{-1}B\), modulo null sets. Then
\[
\mathbf 1_A(U_TY-m)=U_T\bigl(\mathbf 1_B(Y-n)\bigr).
\]
Thus
\[
R_0(\mathbf 1_A(U_TY-m))
=
R_0(\mathbf 1_B(Y-n)).
\]
Consequently, \(m\) is admissible for \(S_t(U_TY)\) in the reconstruction
formula if and only if \(n=m\circ T^{-1}\) is admissible for \(S_t(Y)\). Taking
essential infima gives \(S_t(U_TY)=U_TS_t(Y)\).
\end{proof}

\begin{lemma}[Ergodicity of adapted automorphisms]
\label{lem:adapted-automorphism-ergodicity}
Let \(V\in L^\infty(\F_t)\). If \(U_TV=V\) for every adapted automorphism
\(T\) of the first \(t\) coordinates, then \(V\) is constant.
\end{lemma}

\begin{proof}
For \(t=1\), this is the usual ergodicity of the full group of
measure-preserving automorphisms of \([0,1]\). For \(t>1\), invariance under
measure-preserving transformations of the last coordinate, chosen fiberwise
over \((\pi_1,\dots,\pi_{t-1})\), implies that \(V\) is
\(\F_{t-1}\)-measurable. Repeating the argument backwards reduces to the case
\(t=1\).
\end{proof}

\begin{lemma}[Current-kernel invariance]
\label{lem:current-kernel-invariance}
Let \(R\) be relevant and time consistent, and assume that \(R_0\) is
adapted-law invariant. If \(Y,\widetilde Y\in L^\infty(\F_{t+1})\) satisfy
\[
\law(Y\mid\F_t)=\law(\widetilde Y\mid\F_t),
\]
then
\[
S_t(Y)=S_t(\widetilde Y).
\]
\end{lemma}

\begin{proof}
Set \(M:=\law(Y\mid\F_t)=\law(\widetilde Y\mid\F_t)\). Let
\(m\in L^\infty(\F_t)\) and \(A\in\F_t\). At time \(t\), the conditional laws
of \(\mathbf 1_A(Y-m)\) and \(\mathbf 1_A(\widetilde Y-m)\) are both
\[
\mathbf 1_A(\tau_{-m})_\#M+\mathbf 1_{A^c}\delta_0.
\]
Before time \(t\), the iterated conditional laws are obtained from this same
\(\F_t\)-measurable state. Hence the two terminal variables have the same
adapted law. Since \(R_0\) is adapted-law invariant,
\[
R_0(\mathbf 1_A(Y-m))
=
R_0(\mathbf 1_A(\widetilde Y-m)).
\]
The reconstruction formula therefore gives the same admissible \(m\)'s for
\(Y\) and \(\widetilde Y\), and hence \(S_t(Y)=S_t(\widetilde Y)\).
\end{proof}

\begin{lemma}[One-step identification from \(R_0\)]
\label{lem:one-step-from-R0}
Let \(R=(R_t)_{t=0}^N\) be relevant and time consistent on \(\mathbb U_N\).
Assume that \(R_0\) has the Fatou property and is adapted-law invariant. Then,
for every \(t<N\), there exists a unique relevant static law-invariant risk
measure \(\rho_t:\Pc_b(\R)\to\R\), with the Fatou property, such that
\[
S_t(Y)=\rho_t(\law(Y\mid\F_t))
\qquad\text{a.s.}
\]
for every \(Y\in L^\infty(\F_{t+1})\).
\end{lemma}

\begin{proof}
By \Cref{lem:R0-fatou-implies-one-step-fatou}, \(S_t\) has the Fatou property.

For \(\mu\in\Pc_b(\R)\), let \(q_\mu\) be the quantile function and set
\(Z^\mu:=q_\mu(\pi_{t+1})\). Then \(\law(Z^\mu\mid\F_t)=\mu\). We first show
that \(S_t(Z^\mu)\) is constant. If \(t=0\), this is immediate. If \(t>0\),
then \(Z^\mu\) is fixed by every adapted automorphism of the first
\(t\) coordinates. By \Cref{lem:adapted-automorphism-equivariance},
\(S_t(Z^\mu)\) is fixed by all such automorphisms, and by
\Cref{lem:adapted-automorphism-ergodicity} it is constant. Define
\[
\rho_t(\mu):=S_t(Z^\mu).
\]

The static risk-measure properties of \(\rho_t\) follow from those of \(S_t\).
Normalization is immediate. Cash additivity follows from
\(q_{(\tau_a)_\#\mu}=q_\mu+a\). Monotonicity follows from the quantile
coupling: if bounded random variables with laws \(\mu,\nu\) may be coupled as
\(X\le Y\), then \(q_\mu\le q_\nu\) a.e., and hence
\(\rho_t(\mu)\le\rho_t(\nu)\).

The Fatou property follows from the one-step Fatou property. If
\(\mu_n\weaklyto\mu\) on a common compact support, then
\(q_{\mu_n}\to q_\mu\) a.e. on \((0,1)\), after changing the quantiles on a
null set. Therefore
\begin{align}
\rho_t(\mu)
&=
S_t\big(q_\mu(\pi_{t+1})\big)
\le
\liminf_n S_t\big(q_{\mu_n}(\pi_{t+1})\big)
=
\liminf_n\rho_t(\mu_n).
\label{eq:rho-fatou-from-St}
\end{align}
Thus \(\rho_t\) is lower semicontinuous on bounded-support strata, and in
particular Borel.

We next prove relevance. Let
\(\mu=(1-p)\delta_0+p\delta_\varepsilon\), with \(p>0\) and
\(\varepsilon>0\). Monotonicity gives \(\rho_t(\mu)\ge0\). If
\(\rho_t(\mu)=0\), choose \(B\in\F_t\) with positive probability, taking
\(B=\Omega\) when \(t=0\), and set \(Y:=\mathbf 1_Bq_\mu(\pi_{t+1})\).
Locality gives \(S_t(Y)=0\), so time consistency gives \(R_0(Y)=0\). But
\(Y=\varepsilon\mathbf 1_C\) for a set \(C\) of positive probability, which
contradicts relevance of \(R_0\). Hence \(\rho_t\) is relevant.

It remains to prove
\begin{align}
S_t(Y)=\rho_t(\law(Y\mid\F_t))
\label{eq:one-step-identification-target}
\end{align}
for arbitrary \(Y\in L^\infty(\F_{t+1})\). Set
\(M:=\law(Y\mid\F_t)\), and define the canonical conditional resampling
\(\widehat Y:=q_M(\pi_{t+1})\). Then
\(\law(\widehat Y\mid\F_t)=M\), so
\Cref{lem:current-kernel-invariance} gives \(S_t(Y)=S_t(\widehat Y)\). It is
therefore enough to prove
\begin{align}
S_t(\widehat Y)=\rho_t(M).
\label{eq:one-step-resampling-target}
\end{align}

We prove \eqref{eq:one-step-resampling-target} in three steps. First assume
that \(M\) is finitely valued, say
\(M=\sum_{j=1}^J\mu_j\mathbf 1_{B_j}\), where
\(B_1,\dots,B_J\in\F_t\) form a partition. Then locality gives
\[
S_t(\widehat Y)
=
\sum_j\mathbf 1_{B_j}S_t\big(q_{\mu_j}(\pi_{t+1})\big)
=
\sum_j\mathbf 1_{B_j}\rho_t(\mu_j)
=
\rho_t(M).
\]

Second assume that \(M\) is a.s. supported by a fixed deterministic finite set
\(E=\{x_1<\cdots<x_d\}\). Write
\[
F_i:=M\big(\{x_1,\dots,x_i\}\big),
\qquad i=0,\dots,d,
\]
with \(F_0=0\) and \(F_d=1\). For \(n\ge1\), set
\[
F_i^n:=2^{-n}\lceil 2^nF_i\rceil,
\qquad i=1,\dots,d-1,
\]
and \(F_0^n:=0\), \(F_d^n:=1\). Let \(M^n\) be the random probability measure
on \(E\) defined by
\[
M^n(\{x_i\}) := F_i^n-F_{i-1}^n.
\]
Then \(M^n\) is finitely valued, since all its cumulative weights are dyadic.
Moreover, \(F_i^n\downarrow F_i\), and therefore
\[
q_{M^n}(\pi_{t+1})\uparrow q_M(\pi_{t+1})=\widehat Y
\qquad\text{a.s.}
\]
The finitely valued case gives
\[
S_t\big(q_{M^n}(\pi_{t+1})\big)=\rho_t(M^n).
\]
Using continuity from below for \(S_t\) and for the lift of \(\rho_t\), we get
\[
S_t(\widehat Y)
=
\lim_n S_t\big(q_{M^n}(\pi_{t+1})\big)
=
\lim_n \rho_t(M^n)
=
\rho_t(M).
\]

Finally let \(M\) be arbitrary but bounded. Choose \(K\) such that
\(M(\omega)\in\Pc([-K,K])\) a.s., and define
\(d_n(x):=2^{-n}\lfloor 2^nx\rfloor\). Then \(d_n(x)\uparrow x\) for every
\(x\), and \(\widehat Y_n:=d_n(\widehat Y)\uparrow \widehat Y\). Set
\(M^n:=(d_n)_\#M\). Then \(M^n\) is supported by a deterministic finite grid,
and \(\law(\widehat Y_n\mid\F_t)=M^n\). By
\Cref{lem:current-kernel-invariance}, \(S_t(\widehat Y_n)\) agrees with
\(S_t\big(q_{M^n}(\pi_{t+1})\big)\). Hence the finite-support case gives
\[
S_t(\widehat Y_n)=\rho_t(M^n).
\]
Again by continuity from below,
\[
S_t(\widehat Y)
=
\lim_n S_t(\widehat Y_n)
=
\lim_n \rho_t(M^n)
=
\rho_t(M).
\]
This proves \eqref{eq:one-step-resampling-target}, and hence
\eqref{eq:one-step-identification-target}. Uniqueness follows from
\[
\rho_t(\mu)=S_t\big(q_\mu(\pi_{t+1})\big).
\]
\end{proof}

\begin{proof}[Proof of \Cref{thm:main-standard-cube}]
Assume first that \(R_0\) is adapted-law invariant. By
\Cref{lem:one-step-from-R0}, there are unique relevant static law-invariant
Fatou risk measures \(\rho_t\) such that
\(S_t(Y)=\rho_t(\law(Y\mid\F_t))\). Since \(R\) is time consistent,
\[
R_t(X)=S_t(R_{t+1}(X))
=
\rho_t(\law(R_{t+1}(X)\mid\F_t)).
\]
This proves the one-step representation.

Conversely, assume the one-step representation. Since the \(\rho_t\)'s have
the Fatou property, they are Borel on bounded-support strata. Define
recursively
\begin{align}
r_N(x)&:=x,
&
r_t(m)&:=\rho_t\big((r_{t+1})_\#m\big),
\qquad m\in\Pc_b^{N-t}(\R).
\label{eq:r-recursion-main-proof}
\end{align}
The maps \(r_t\) are Borel on bounded-support strata by backward induction.
If \(L_N(X):=X\) and \(L_t(X):=\law(L_{t+1}(X)\mid\F_t)\), then
\(L_0(X)=\lawad(X)\), and a backward induction using
\eqref{eq:r-recursion-main-proof} gives
\[
R_t(X)=r_t(L_t(X)).
\]
In particular,
\[
R_0(X)=r_0(\lawad(X)).
\]
Thus \(R_0\) factors through the adapted law, proving adapted-law invariance.

We now prove the converse construction. Let
\(\rho_0,\dots,\rho_{N-1}\) be relevant static law-invariant risk measures with
the Fatou property, and define \(R\) by the backward recursion
\begin{align}
R_N(X)&=X,
&
R_t(X)&=\rho_t(\law(R_{t+1}(X)\mid\F_t)).
\label{eq:constructed-backward-recursion}
\end{align}
By \Cref{lem:main-conditional-lifts}, the one-step maps are normalized,
monotone, cash additive, and have the Fatou property. Hence the backward
composition is a time-consistent dynamic risk measure. The preceding
\(r_t\)-recursion shows that \(R_0\) factors through \(\lawad\).

Relevance follows by backward induction. If \(X=\varepsilon\mathbf 1_A\) with
\(\P(A)>0\), set \(X_N:=X\) and \(X_t:=R_t(X)\). Then \(X_t\ge0\) for all
\(t\). If \(X_{t+1}>0\) with positive probability, then on the set
\[
B_t:=\{\P(X_{t+1}>0\mid\F_t)>0\}
\]
the conditional law of \(X_{t+1}\) gives positive mass to \((0,\infty)\).
Relevance and monotonicity of \(\rho_t\) imply \(X_t>0\) on \(B_t\). Hence
\(R_0(X)>0\).

The Fatou property of the constructed \(R_0\) follows from the Fatou
properties of the one-step maps. First use
\Cref{lem:main-continuity-from-below} backward in time to get continuity from
below for the whole composition. If \(X_n\to X\) a.s. with a common bound, put
\(Y_k:=\inf_{n\ge k}X_n\). Then \(Y_k\uparrow X\) and \(Y_k\le X_n\) for
\(n\ge k\), so
\[
R_0(X)=\lim_kR_0(Y_k)\le\liminf_n R_0(X_n).
\]

Finally, convexity and coherence are inherited one step at a time. If \(R\) is
convex, then each \(S_t\) is conditionally convex; testing on one-step
realizations of arbitrary bounded joint laws on the fresh coordinate
\(\pi_{t+1}\) gives convexity of \(\rho_t\). Conversely, if each \(\rho_t\) is
convex, then each conditional lift is conditionally convex by
\Cref{lem:main-conditional-lifts}, and the monotone backward composition of
conditionally convex maps is conditionally convex. The proof for coherence is
the same, adding positive homogeneity.
\end{proof}

\section{Adapted law invariance for cash-flow processes}
\label{sec:cash-flow-processes}

The preceding results concern a single terminal liability. We briefly record
the corresponding cash-flow version of the main characterization. To avoid
num\'eraire issues, we assume deterministic interest rates and express all
payments in units of the bank account. Thus
\(C=(C_1,\dots,C_N)\), with \(C_t\in L^\infty(\F_t)\), is interpreted as a
stream of discounted losses paid at the respective dates.

For \(t=0,\dots,N\), write
\[
\mathcal C_t:=\prod_{s=t+1}^N L^\infty(\F_s),
\qquad
\mathcal C_N:=\{0\},
\]
and let \(V_t:\mathcal C_t\to L^\infty(\F_t)\) denote the time-\(t\)
assessment of the remaining cash-flow stream. A dynamic risk measure on processes is the family \(V=(V_t)_{t=0}^N\). We write
\(C_{t+1:N}:=(C_{t+1},\dots,C_N)\), with empty tails identified with zero.
The full time-\(t\) assessment of a stream including a current payment
\(C_t\) is \(C_t+V_t(C_{t+1:N})\), so current-date cash additivity is
built into this notation. Following the composition framework
for risk measures on processes, see
\cite{CheriditoKupper2011,AcciaioFollmerPenner2012}, put
\[
H_t(Y):=V_t(Y,0,\dots,0),
\qquad Y\in L^\infty(\F_{t+1}),
\]
and assume that the maps \(H_t\) are normalized and monotone and that time
consistency takes the recursive form
\begin{align}
V_t(C_{t+1:N})
=
H_t\bigl(C_{t+1}+V_{t+1}(C_{t+2:N})\bigr),
\qquad t=0,\dots,N-1.
\label{eq:cash-flow-time-consistency}
\end{align}

The definition of \(\lawad\) extends verbatim to bounded
\(\R^N\)-valued terminal random variables. We use this extension for the
adapted cash-flow vector \(C=(C_1,\dots,C_N)\), and call \(V\)
\emph{adapted-law invariant} if
\[
\lawad(C)=\lawad(\widetilde C)
\quad\Longrightarrow\quad
V_0(C_{1:N})=V_0(\widetilde C_{1:N}).
\]

\begin{definition}[Conditional one-step dated cash additivity]
\label{def:one-step-dated-cash-additivity}
The cash-flow risk measure \(V\) is \emph{conditionally one-step dated cash
additive} if, for every \(t<N\), \(Y\in L^\infty(\F_{t+1})\), and
\(m\in L^\infty(\F_t)\),
\begin{align}
H_t(Y+m)=H_t(Y)+m
\qquad\text{a.s.}
\label{eq:one-step-dated-cash-additivity}
\end{align}
\end{definition}

Since all amounts are expressed in discounted units,
\eqref{eq:one-step-dated-cash-additivity} says that an
\(\F_t\)-measurable loss paid one period later raises the time-\(t\)
assessment by exactly its discounted amount.

\begin{corollary}[Cash-flow version of the main characterization]
\label{cor:cash-flow-main-characterization}
Let \(V=(V_t)_{t=0}^N\) satisfy
\eqref{eq:cash-flow-time-consistency} on the standard filtered cube, and
assume that it is conditionally one-step dated cash additive. Suppose that
\(V_0\) is relevant, in the sense that
\(V_0(C_{1:N})>0\) for every componentwise nonnegative cash-flow stream with
\(\P(\sum_{s=1}^N C_s>0)>0\), and that \(V_0\) has the Fatou property under
uniformly bounded componentwise a.s. convergence. Then the following are
equivalent.

\begin{enumerate}[(i)]
\item \(V\) is adapted-law invariant.

\item There are unique relevant static law-invariant risk measures
\(\rho_t:\Pc_b(\R)\to\R\), \(t=0,\dots,N-1\), with the Fatou property, such
that
\begin{align}
V_N(0)&=0,
&
V_t(C_{t+1:N})
&=
\rho_t\!\left(
\law\bigl(C_{t+1}+V_{t+1}(C_{t+2:N})\mid\F_t\bigr)
\right),
\quad t=N-1,\dots,0.
\label{eq:cash-flow-conditional-law-recursion}
\end{align}
\end{enumerate}

Here convexity and coherence of \(V\) are understood componentwise. Under
these equivalent conditions, \(V\) is convex if and only if each \(\rho_t\)
is convex, and \(V\) is coherent if and only if each \(\rho_t\) is coherent.
Conversely, any family \(\rho_0,\dots,\rho_{N-1}\) as in (ii)
defines through \eqref{eq:cash-flow-conditional-law-recursion} a relevant
time-consistent conditionally one-step dated cash-additive cash-flow risk
measure whose initial functional has the Fatou property and is adapted-law
invariant.
\end{corollary}

\begin{proof}
Define a family on terminal positions by
\[
R_N(X):=X,
\qquad
R_t(X):=H_t(R_{t+1}(X)),
\quad t=N-1,\dots,0.
\]
By \eqref{eq:one-step-dated-cash-additivity}, each \(H_t\) is cash additive.
Hence \(R\) is a time-consistent dynamic risk measure in the sense of
\Cref{def:dynamic-risk,def:dynamic-properties}, and \(H_t\) is its one-step
restriction.

A backward induction gives the aggregation identity
\begin{align}
V_t(C_{t+1:N})
=
R_t\!\left(\sum_{s=t+1}^N C_s\right),
\qquad t=0,\dots,N.
\label{eq:cash-flow-aggregation}
\end{align}
Indeed, if the identity holds at time \(t+1\), then cash additivity of
\(R_{t+1}\) gives
\[
R_{t+1}\!\left(\sum_{s=t+1}^N C_s\right)
=
C_{t+1}
+
R_{t+1}\!\left(\sum_{s=t+2}^N C_s\right)
=
C_{t+1}+V_{t+1}(C_{t+2:N}),
\]
and applying \(H_t\) yields the claim at time \(t\).

By functoriality of iterated conditional laws under deterministic Borel maps,
equality of the adapted laws of two cash-flow vectors implies equality of the
adapted laws of their sums.\footnote{More precisely, if $f\to F$ is a Borel map between standard Borel spaces, then
$\law(f(Z)\mid\G)=f_\#\law(Z\mid\G)$ 
a.s. Iterating this identity shows that
$\lawad(f(Z))$ is determined by $\lawad(Z)$. 
Hence equality of adapted laws
is preserved under deterministic Borel transformations. 
Here we apply this to the summation map.} 
Conversely, the terminal embedding
\(X\mapsto(0,\dots,0,X)\) preserves equality of adapted laws. In view of
\eqref{eq:cash-flow-aggregation}, \(V_0\) is therefore adapted-law invariant
if and only if \(R_0\) is adapted-law invariant. The relevance and Fatou
assumptions on \(V_0\) imply the corresponding properties of \(R_0\) by
restriction to terminal-only streams.

The result now follows from \Cref{thm:main-standard-cube}: there are unique
static law-invariant risk measures \(\rho_t\) such that
\(H_t(Y)=\rho_t(\law(Y\mid\F_t))\). Substituting
\(Y=C_{t+1}+V_{t+1}(C_{t+2:N})\) in
\eqref{eq:cash-flow-time-consistency} gives
\eqref{eq:cash-flow-conditional-law-recursion}. The converse, as well as the
claims concerning convexity and coherence, follows from the converse part of
\Cref{thm:main-standard-cube} together with
\eqref{eq:cash-flow-aggregation}.
\end{proof}

Thus, under conditional one-step dated cash additivity, the cash-flow problem
contains no additional structure beyond the aggregate discounted loss in
\eqref{eq:cash-flow-aggregation}.

\begin{remark}[The cash-subadditive case]
\label{rem:cash-flow-subadditive}
In the general process framework, one may impose only the one-sided condition
\begin{align}
H_t(Y+m)\le H_t(Y)+m,
\qquad m\in L^\infty_+(\F_t),
\label{eq:one-step-dated-cash-subadditivity}
\end{align}
which is the cash-subadditive analogue of
\eqref{eq:one-step-dated-cash-additivity} in the present loss convention; see
\cite{CheriditoKupper2011,AcciaioFollmerPenner2012}. Then the family obtained by
composing the \(H_t\)'s need not be cash additive, and the aggregation identity
\eqref{eq:cash-flow-aggregation} generally fails. Thus payment dates may remain
relevant even for cash-flow streams with the same total discounted loss.

Consequently, the cash-subadditive case does not follow from
\Cref{thm:main-standard-cube}. A corresponding characterization would require
a process-level reconstruction and one-step identification theorem. We leave this to future research. 
\end{remark}


\section{Proof of the two-period Kupper--Schachermayer theorem}
\label{sec:KS-two-period-proof}

We prove the two-period version of \Cref{thm:KS-finite-rigidity}. 
Recall that
\[
\mathbb U_2
=
([0,1]^2,\mathcal B([0,1]^2),\lambda^2,(\F_t)_{t=0}^2),
\qquad
\F_0=0,\quad \F_1=\sigma(\pi_1),\quad \F_2=\sigma(\pi_1,\pi_2).
\]
For \(\gamma\in\R\setminus\{0\}\), define the law-level functional
\begin{align}
\mathrm e^\gamma(\mu)
&:=
\frac1\gamma\log\int_\R e^{\gamma x}\,\mu(dx),
\label{eq:KS-law-entropic}
\end{align}
with limiting cases
\begin{align}
\mathrm e^0(\mu):=\int_\R x\,\mu(dx),
\qquad
\mathrm e^\infty(\mu):=\sup\supp(\mu).
\label{eq:KS-law-entropic-limits}
\end{align}
The corresponding conditional operators are
\begin{align}
\mathcal E_t^\gamma(X)
:=
\mathrm e^\gamma(\law(X\mid\F_t)),
\qquad t=0,1,2.
\label{eq:KS-conditional-entropic}
\end{align}

The probabilistic part of the proof is encoded in the next two lemmas.

\begin{lemma}[Independent future risks are evaluated deterministically, see {\cite[Lemma 2.2]{KupperSchachermayer2009}}]
\label{lem:KS-independent-two-period}
Let \(R=(R_0,R_1,R_2)\) be a relevant time-consistent dynamic risk measure on
\(\mathbb U_2\), and assume that \(R_0\) is terminal-law invariant. If
\(Y\in L^\infty(\F_2)\) is independent of \(\F_1\), then
\[
R_1(Y)=R_0(Y)
\qquad\text{a.s.}
\]
\end{lemma}

\begin{proof}
Put \(V:=R_1(Y)\). We first show that \(V\) is constant. If not, there are
\(a<b\) such that \(\P(V\le a)>0\) and \(\P(V\ge b)>0\). Since \(\F_1\) is
atomless, choose \(A\subset\{V\ge b\}\) and \(A'\subset\{V\le a\}\), both in
\(\F_1\), with \(\P(A)=\P(A')>0\).

By locality and cash additivity,
\begin{align}
R_1(\mathbf 1_A(Y-a))&=\mathbf 1_A(V-a),
&
R_1(\mathbf 1_{A'}(Y-a))&=\mathbf 1_{A'}(V-a).
\label{eq:KS-independent-locality}
\end{align}
Time consistency gives the same identities after applying \(R_0\). Since \(Y\)
is independent of \(\F_1\) and \(\P(A)=\P(A')\), the variables
\(\mathbf 1_A(Y-a)\) and \(\mathbf 1_{A'}(Y-a)\) have the same law. Hence,
by terminal-law invariance, the corresponding \(R_0\)-values are equal. But
\[
\mathbf 1_A(V-a)\ge (b-a)\mathbf 1_A,
\qquad
\mathbf 1_{A'}(V-a)\le0,
\]
so monotonicity and relevance imply
\[
R_0(\mathbf 1_A(V-a))>0,
\qquad
R_0(\mathbf 1_{A'}(V-a))\le0,
\]
a contradiction. Thus \(V\) is constant. Finally,
\(R_0(Y)=R_0(R_1(Y))=R_0(V)=V\).
\end{proof}

\begin{lemma}[Decomposability of the initial law functional]
\label{lem:KS-initial-decomposable-two-period}
Let \(R=(R_0,R_1,R_2)\) satisfy the assumptions of
\Cref{lem:KS-independent-two-period}, and let \(c:\Pc_b(\R)\to\R\) be the
law-level representative of \(R_0\), that is, \(c(\law(X)):=R_0(X)\). Then
\begin{align}
c\Big(\sum_i p_i\mu_i\Big)
=
c\Big(\sum_i p_i\delta_{c(\mu_i)}\Big)
\label{eq:KS-c-decomposable}
\end{align}
for every probability vector \((p_i)_{i=1}^k\), with zero weights omitted, and
all \(\mu_i\in\Pc_b(\R)\).
\end{lemma}

\begin{proof}
Choose a partition \(A_1,\dots,A_k\in\F_1\) with \(\P(A_i)=p_i\), and set
\(Y_i:=q_{\mu_i}(\pi_2)\). Then \(Y_i\) is independent of \(\F_1\) and has
law \(\mu_i\). For \(Y:=\sum_i\mathbf 1_{A_i}Y_i\), one has
\(\law(Y)=\sum_i p_i\mu_i\). By locality and
\Cref{lem:KS-independent-two-period},
\[
R_1(Y)
=
\sum_i\mathbf 1_{A_i}R_1(Y_i)
=
\sum_i\mathbf 1_{A_i}c(\mu_i).
\]
Thus \(\law(R_1(Y))=\sum_i p_i\delta_{c(\mu_i)}\), and time consistency gives
\eqref{eq:KS-c-decomposable}.
\end{proof}

The rest of the proof is deterministic. We use the following standard
finite-lottery form of the Kolmogorov--Nagumo--de Finetti \cite{Kolmogorov1930,Nagumo1930,deFinetti1931} characterization of
quasi-arithmetic means; see
 \cite[Chapter V]{Aczel1966}.

\begin{theorem}[Kolmogorov--Nagumo--de Finetti representation]
\label{thm:finite-vNM-deFinetti}
Let \(I\subset\R\) be an interval, and let \(C\) be a certainty equivalent on
finitely supported probability measures on \(I\). 


Assume that \(C\) is normalized on constants,
internal, strictly increasing, mixture-continuous, and satisfies reduction of
compound lotteries:
\begin{align}
\min\supp(\mu)\le C(\mu)\le \max\supp(\mu),
\qquad
C(\delta_x)=x,
\label{eq:KNdF-internality}
\\
C\Big(\sum_i p_i\mu_i\Big)
=
C\Big(\sum_i p_i\delta_{C(\mu_i)}\Big).
\label{eq:KNdF-reduction}
\end{align}
Here internal signifies that
\(
\min\supp(\mu)\le C(\mu)\le \max\supp(\mu).
\)
Strict monotonicity means that, if \(x<y\) and \(p\in(0,1)\), then
\(x<C(p\delta_y+(1-p)\delta_x)<y\), and mixture-continuity means that
\(p\mapsto C(p\mu+(1-p)\nu)\) is continuous for finitely supported
\(\mu,\nu\). Then there exists a continuous strictly monotone function
\(u:I\to\R\), unique up to nonzero affine transformations, such that
\begin{align}
C\Big(\sum_i p_i\delta_{x_i}\Big)
=
u^{-1}\Big(\sum_i p_i u(x_i)\Big).
\label{eq:KNdF-representation}
\end{align}
\end{theorem}

\begin{lemma}[Decomposable finite-lottery certainty equivalents]
\label{lem:KS-finite-lottery-reduction}
Let \(c\) be defined on finitely supported laws and assume that it is
normalized, monotone, cash additive, relevant, and decomposable. Then either
\(c(\mu)=\max\supp(\mu)\) for every finitely supported law \(\mu\), or there
exists a continuous strictly monotone function \(u:\R\to\R\) such that
\begin{align}
c\Big(\sum_i p_i\delta_{x_i}\Big)
=
u^{-1}\Big(\sum_i p_i u(x_i)\Big)
\label{eq:KS-quasi-arithmetic}
\end{align}
for every finitely supported law \(\sum_i p_i\delta_{x_i}\).
\end{lemma}

\begin{proof}
For \(p\in[0,1]\), write
\(M_p(x,y):=c(p\delta_x+(1-p)\delta_y)\). By monotonicity and normalization,
\(M_p\) is internal; by cash additivity it is translation invariant. Moreover,
monotonicity and cash additivity give
\begin{align}
|M_p(x,y)-M_p(x',y')|
\le
\max\{|x-x'|,|y-y'|\}.
\label{eq:KS-binary-lipschitz}
\end{align}
Decomposability gives the binary reduction identity
\begin{align}
M_{rp+(1-r)q}(b,a)
=
M_r(M_p(b,a),M_q(b,a)),
\qquad a<b.
\label{eq:KS-binary-reduction}
\end{align}


We first separate endpoint cases. Fix \(a<b\) and put \(f(p):=M_p(b,a)\).
If \(f(p_0)=b\) for some \(p_0\in(0,1)\), then
\[
f(p_0r)=M_r(f(p_0),a)=M_r(b,a)=f(r).
\]
Together with monotonicity in \(p\), this gives \(f(q)=b\) for every
\(q>0\). After translation, \(d:=b-a\) therefore satisfies
\[
M_p(d,0)=d,\qquad p\in(0,1].
\]

We claim that this already implies \(c=\max\supp\) on all finite lotteries.
Let
\[
H:=\{h>0: M_p(h,0)=h\text{ for all }p\in(0,1]\}.
\]
Then \(d\in H\). If \(h\in H\) and \(0<r\le h\), then
\[
h=M_p(h,0)\le M_p(h,h-r)\le h,
\]
and cash additivity gives \(M_p(r,0)=r\). Hence \(H\) is downward closed.

Moreover, if \((0,h]\subset H\), then \((0,2h]\subset H\). Indeed, let
\(L=h+e\), \(0<e\le h\). For \(\alpha,q\in(0,1)\), set
\[
\mu:=\alpha(q\delta_L+(1-q)\delta_e)+(1-\alpha)\delta_0 .
\]
Since \(L-e=h\in H\), decomposability gives
\[
c(\mu)=M_\alpha(L,0).
\]
Writing also
\[
\mu=\alpha q\,\delta_L+(1-\alpha q)
\left(
\frac{\alpha(1-q)}{1-\alpha q}\delta_e
+
\frac{1-\alpha}{1-\alpha q}\delta_0
\right),
\]
and using \(e\in H\), decomposability gives \(c(\mu)=M_{\alpha q}(L,e)=L\).
Thus \(M_\alpha(L,0)=L\), and \(L\in H\). By iteration, \(H=(0,\infty)\).

Hence every binary lottery with positive mass at its upper outcome is evaluated
at that upper outcome. If \(\mu\) is finitely supported, \(y:=\max\supp(\mu)\),
and \(p:=\mu(\{y\})>0\), write \(\mu=p\delta_y+(1-p)\nu\). Since
\(m:=c(\nu)<y\), decomposability gives
\[
c(\mu)=M_p(y,m)=y.
\]
Therefore \(c=\max\supp\) on all finite lotteries.

The lower endpoint case is excluded by relevance. Indeed, if \(M_p(b,a)=a\)
for some \(a<b\) and \(p\in(0,1)\), then, after translating by \(-a\),
\(c((1-p)\delta_0+p\delta_{b-a})=0\), contradicting relevance.

We may therefore assume the strictly internal case. Then \(f\) is strictly
increasing: for \(0<p<q<1\),
\(f(p)=M_{p/q}(f(q),a)<f(q)\). It is also continuous. Since \(f\) is monotone,
left and right limits exist. If \(f(p-)<f(p)\), choose \(p_n\uparrow p\);
applying \eqref{eq:KS-binary-reduction} to \((p_n+p)/2\) and using
\eqref{eq:KS-binary-lipschitz} gives
\(f(p-)=M_{1/2}(f(p-),f(p))\), contradicting strict internality. Right
continuity is identical.

The same argument at the endpoints gives \(f(0+)=a\) and \(f(1-)=b\), so
\(p\mapsto M_p(b,a)\) is continuous on \([0,1]\). Hence \(c\) is
mixture-continuous. Indeed, for finitely supported \(\mu,\nu\),
decomposability yields
\[
c(p\mu+(1-p)\nu)
=
c\bigl(p\delta_{c(\mu)}+(1-p)\delta_{c(\nu)}\bigr),
\]
and the right-hand side is either constant or one of the continuous binary
maps just considered.

Thus, on each compact interval \(I=[a,b]\), the restriction of \(c\) satisfies
the hypotheses of \Cref{thm:finite-vNM-deFinetti}. Hence
\eqref{eq:KNdF-representation} holds on \(I\). By uniqueness in
\Cref{thm:finite-vNM-deFinetti}, the local generators on overlapping intervals
are nonzero affine transforms of one another. Patching them gives a global
continuous strictly monotone \(u:\R\to\R\), and hence
\eqref{eq:KS-quasi-arithmetic}.
\end{proof}

\begin{lemma}[Cash-additive quasi-arithmetic means]
\label{lem:KS-cash-additive-quasi-arithmetic}
Let \(u:\R\to\R\) be continuous and strictly monotone, and assume that
\(c\) is given by \eqref{eq:KS-quasi-arithmetic}. If \(c\) is cash additive,
then either \(u\) is affine, in which case \(c\) is the expectation, or there
exists \(\gamma\in\R\setminus\{0\}\) such that
\begin{align}
c\Big(\sum_i p_i\delta_{x_i}\Big)
=
\frac1\gamma\log\sum_i p_i e^{\gamma x_i}.
\label{eq:KS-entropic-finite}
\end{align}
\end{lemma}

\begin{proof}
Let \(I:=u(\R)\). For \(m\in\R\), define
\(T_m:I\to I\) by \(T_m(z):=u(u^{-1}(z)+m)\). Cash additivity says that
\(T_m\) preserves finite convex combinations. Indeed, for \(z_i=u(x_i)\),
\begin{align}
u^{-1}\Big(\sum_i p_iT_m(z_i)\Big)
&=
c\Big(\sum_i p_i\delta_{x_i+m}\Big)
=
c\Big(\sum_i p_i\delta_{x_i}\Big)+m
\nonumber\\
&=
u^{-1}\Big(T_m\big(\sum_i p_iz_i\big)\Big).
\label{eq:KS-Tm-affine-reason}
\end{align}
Thus \(T_m\) is continuous and affine on the interval \(I\):
\(T_m(z)=a(m)z+b(m)\), with \(a(m)>0\). The identity
\(T_{m+n}=T_m\circ T_n\) gives
\begin{align}
a(m+n)=a(m)a(n),
\qquad
b(m+n)=a(m)b(n)+b(m).
\label{eq:KS-affine-cocycle}
\end{align}
By continuity in \(m\), either \(a\equiv1\), or \(a(m)=e^{\gamma m}\) for
some \(\gamma\ne0\).

If \(a\equiv1\), then \(b(m+n)=b(m)+b(n)\), hence \(b(m)=\beta m\) with
\(\beta\ne0\). Thus \(u(x+m)=u(x)+\beta m\), so \(u\) is affine and \(c\) is
the expectation.

If \(a(m)=e^{\gamma m}\), then commutativity of translations implies
\(b(m)=\kappa(e^{\gamma m}-1)\) for some \(\kappa\). Hence \(v:=u+\kappa\)
satisfies \(v(x+m)=e^{\gamma m}v(x)\). Thus \(v(x)=v(0)e^{\gamma x}\), with
\(v(0)\ne0\). Since quasi-arithmetic means are unchanged by nonzero affine
transformations of the generator, \eqref{eq:KS-entropic-finite} follows.
\end{proof}

\begin{lemma}[Deterministic classification]
\label{lem:KS-deterministic-classification}
Let \(c:\Pc_b(\R)\to\R\) be normalized, monotone, cash additive, relevant, and
decomposable. Then \(c=\mathrm e^\gamma\) for some
\(\gamma\in(-\infty,\infty]\).
\end{lemma}

\begin{proof}
On finitely supported laws, \Cref{lem:KS-finite-lottery-reduction} gives
either the upper endpoint or the quasi-arithmetic form
\eqref{eq:KS-quasi-arithmetic}. In the latter case,
\Cref{lem:KS-cash-additive-quasi-arithmetic} gives either expectation or
\eqref{eq:KS-entropic-finite}.

It remains to pass from finite laws to bounded laws. Monotonicity and cash
additivity imply
\[
|c(\law(X))-c(\law(Y))|\le \|X-Y\|_\infty
\]
whenever \(X,Y\) are realized on a common probability space. Uniform
approximation by finite-valued random variables extends the formula to all of
\(\Pc_b(\R)\).
\end{proof}

\begin{theorem}[Two-period Kupper--Schachermayer rigidity]
\label{thm:KS-two-period-rigidity}
Let \(R=(R_0,R_1,R_2)\) be a relevant time-consistent dynamic risk measure on
\(\mathbb U_2\). If \(R\) is terminal-law invariant, then there exists
\(\gamma\in(-\infty,\infty]\) such that
\[
R_t(X)=\mathcal E_t^\gamma(X),
\qquad t=0,1,2.
\]
Conversely, every such family is relevant, time consistent, and
terminal-law invariant. Moreover, it is convex if and only if
\(\gamma\in[0,\infty]\), and coherent if and only if
\(\gamma\in\{0,\infty\}\).
\end{theorem}

\begin{proof}
Let \(c(\law(X)):=R_0(X)\). By
\Cref{lem:KS-initial-decomposable-two-period}, \(c\) satisfies
\eqref{eq:KS-c-decomposable}. It is normalized, monotone, cash additive, and
relevant because \(R_0\) has these properties. Hence
\Cref{lem:KS-deterministic-classification} gives \(c=\mathrm e^\gamma\) for
some \(\gamma\in(-\infty,\infty]\). Thus
\[
R_0(X)=\mathcal E_0^\gamma(X).
\]

The rest follows from reconstruction. The family
\((\mathcal E_t^\gamma)_{t=0}^2\) is a relevant time-consistent dynamic risk
measure with the same initial functional \(R_0\), by the tower property for
finite \(\gamma\), and by the corresponding tower properties of conditional
expectation and conditional essential supremum in the limiting cases. Since a
relevant time-consistent dynamic risk measure is uniquely determined by its
initial functional, \(R_t=\mathcal E_t^\gamma\) for \(t=0,1,2\).

The converse direction is the same tower-property check. The convexity and
coherence assertions are standard: \(\gamma>0\) gives convexity by conditional
Hölder, \(\gamma=0\) is linear, \(\gamma=\infty\) is sublinear, while
\(\gamma<0\) is not convex on an atomless space. Positive homogeneity holds
only for \(\gamma=0\) and \(\gamma=\infty\).
\end{proof}

\section{Proof of the adapted Kusuoka representation}
\label{sec:adapted-kusuoka}

We use the following standard form of the static Kusuoka theorem in the loss
convention of this paper. We state it as a result on the base functional, the original statement \cite[Theorem 4]{Kusuoka2001} involves the lifted functional. The translation from one statement to the other uses the fact that any probability measure on $\R$ can be represented as the law of a random variable defined over an atomless probability space.

\begin{theorem}[Static Kusuoka representation]
\label{thm:static-kusuoka}
A static law-invariant risk measure \(\rho:\Pc_b(\R)\to\R\) is coherent and
has the Fatou property if and only if there exists a nonempty convex weakly
closed set \(\mathcal M\subset\Pc([0,1])\) such that
\[
\rho(\mu)
=
\sup_{\nu\in\mathcal M}
\int_{[0,1]}\AVaR_\alpha(\mu)\,\nu(d\alpha),
\qquad
\mu\in\Pc_b(\R).
\]
\end{theorem}

\begin{proof}[Proof of \Cref{cor:adapted-kusuoka}]
Assume first that \(R\) is adapted-law invariant and coherent. By
\Cref{thm:main-standard-cube}, the one-step maps have the form
\[
S_t(Y)=\rho_t(\law(Y\mid\F_t)),
\]
where each \(\rho_t:\Pc_b(\R)\to\R\) is a relevant static law-invariant risk
measure with the Fatou property. The same theorem gives coherence of each
\(\rho_t\). Applying \Cref{thm:static-kusuoka} to \(\rho_t\) yields a nonempty
convex weakly closed set \(\mathcal M_t\subset\Pc([0,1])\) such that
\[
\rho_t(\mu)
=
\sup_{\nu\in\mathcal M_t}
\int_{[0,1]}\AVaR_\alpha(\mu)\,\nu(d\alpha).
\]
Substituting \(\mu=\law(Y\mid\F_t)\) gives the conditional one-step formula.
The nested formula follows from time consistency,
\(R_t=S_t\circ S_{t+1}\circ\cdots\circ S_{N-1}\).

Conversely, let nonempty convex weakly closed sets
\(\mathcal M_t\subset\Pc([0,1])\) be given and define
\[
\rho_t(\mu)
:=
\sup_{\nu\in\mathcal M_t}
\int_{[0,1]}\AVaR_\alpha(\mu)\,\nu(d\alpha).
\]
By \Cref{thm:static-kusuoka}, each \(\rho_t\) is static law invariant,
coherent, and has the Fatou property. It is also relevant: if
\(\mu=(1-p)\delta_0+p\delta_\varepsilon\) with \(p>0\) and
\(\varepsilon>0\), then \(\AVaR_\alpha(\mu)>0\) for every
\(\alpha\in[0,1]\), hence \(\rho_t(\mu)>0\).

Now define \(R\) by the backward recursion
\[
R_N(X):=X,
\qquad
R_t(X):=\rho_t(\law(R_{t+1}(X)\mid\F_t)).
\]
The converse part of \Cref{thm:main-standard-cube} gives relevance, time
consistency, adapted law invariance, and the Fatou property of \(R_0\).
Coherence follows either from \Cref{thm:main-standard-cube} or directly from
the conditional coherence of the one-step lifts and monotone backward
composition. This proves the converse implication and the final assertion.
\end{proof}

{\bf Acknowledgment:}
MB gratefully acknowledges financial support from FWF through P34743, P35197, and FW506064 and from OeNB through AB1898311. SP gratefully acknowledges the financial support from the Natural Sciences and Engineering Research Council (NSERC) of Canada with grant number RGPIN-2025-05847. 
MS gratefully acknowledges financial support from FWF through FW506064.
This research was funded in whole or in part by the Austrian Science Fund (FWF) [10.55776/P34743,  10.55776/P35197, and 10.55776/FW506064]. For open access purposes, the authors have applied a CC BY public copyright license to any author accepted manuscript version arising from this submission.

\printbibliography
\end{document}